\documentclass[US-letter,USenglish,authorcolumns,autoref]{lipics-v2021}

\usepackage[noend]{algpseudocode}
\usepackage{algorithm}
\usepackage{times}
\usepackage{amssymb}
\usepackage{amsmath}
\usepackage{latexsym}
\usepackage{graphics}
\usepackage{url}
\usepackage{verbatim}
\usepackage{color}
\usepackage{setspace}
\usepackage{multirow}
\usepackage{lineno}
\usepackage{booktabs}
\usepackage{graphicx}
\usepackage{caption}
\usepackage[shortlabels]{enumitem}

\renewcommand{\tt} {\mathtt}

\definecolor{winered}{rgb}{0.5,0.2,0}
\definecolor{lipicsYellow}{rgb}{0.99,0.78,0.07}

\usepackage{hyperref}
\hypersetup{ 
    colorlinks = true,
    allcolors={winered},
}

\usepackage{wrapfig}
\usepackage[normalem]{ulem}
\usepackage{changepage}

\algnewcommand\And{\textbf{and}}
\algnewcommand\Or{\textbf{or}}
\algnewcommand\Not{\textbf{not}}
\algnewcommand\In{\textbf{in}}
\algnewcommand\Each{\textbf{each}}

\newcommand{\squishlist}{
 \begin{list}{$\bullet$}
  { \setlength{\itemsep}{0pt}
     \setlength{\parsep}{3pt}
     \setlength{\topsep}{3pt}
     \setlength{\partopsep}{0pt}
     \setlength{\leftmargin}{2.5em}
     \setlength{\labelwidth}{1em}
     \setlength{\labelsep}{0.5em} } }

\newcommand{\squishlisttwo}{
 \begin{list}{$\triangleright$}
  { \setlength{\itemsep}{0pt}
     \setlength{\parsep}{0pt}
    \setlength{\topsep}{0pt}
    \setlength{\partopsep}{0pt}
    \setlength{\leftmargin}{2em}
    \setlength{\labelwidth}{1.5em}
    \setlength{\labelsep}{0.5em} } }

\newcommand{\squishend}{
  \end{list}  }

\usepackage{listings}

\definecolor{verbgray}{gray}{0.9}

\lstnewenvironment{code}{%
  \lstset{backgroundcolor=\color{verbgray},
 frame=single,
  language=C,
  framerule=0pt,
  basicstyle=\ttfamily,
  showstringspaces=false,
  commentstyle=\color{blue}\textit,
  keywordstyle=\color{black}\bf,
  numbers=none,
  keepspaces=true,
  columns=fullflexible}}{}

\newcommand{\bS}{\mathbf{S}}

\newcommand{\bM}{\mathbf{M}}

\newcommand{\bN}{\mathbf{N}}

\definecolor{shadecolor}{rgb}{.91, .91, .91}
\definecolor{bordercolor}{rgb}{.8, .8, .6}

\definecolor{ultramarine}{rgb}{0, 0.125, 0.376}

 \definecolor{arsenic}{rgb}{0.23, 0.27, 0.29}
 \definecolor{beige}{rgb}{0.96, 0.96, 0.86}

\definecolor{amber}{rgb}{1.0, 0.75, 0.0}
\definecolor{orange}{rgb}{1.0, 0.49, 0.0}
\definecolor{dandelion}{rgb}{0.94, 0.88, 0.19}

  \definecolor{indiagreen}{rgb}{0.07, 0.53, 0.03}
  \definecolor{huntergreen}{rgb}{0.21, 0.37, 0.23}

\newcommand{\blue}[1] {\textcolor{blue}{#1}}

\newcommand{\defo}[1] {\emph{\textcolor{blue}{#1}}}

\definecolor{shadecolor}{rgb}{.9, .9, .9}

 \usepackage{framed}

    {\endMakeFramed}

    \newenvironment{frshaded*}{%
    \MakeFramed {\advance\hsize-\width \FrameRestore}}%
    {\endMakeFramed}
    
    \newcounter{examplecounter}
\newenvironment{exam}{
 \begin{frshaded*}
    \refstepcounter{examplecounter}%
    \noindent
  \textbf{Example \arabic{examplecounter}}%
  \quad
}{%
\end{frshaded*}
}

{\endMakeFramed}

\newenvironment{frshaded2*}{%
    \MakeFramed {\advance\hsize-\width \FrameRestore}}%
    {\endMakeFramed}

\newenvironment{result}{
 \begin{frshaded2*}
}{%
\end{frshaded2*}

}

{\endMakeFramed}

\newenvironment{frshaded3*}{%
    \MakeFramed {\advance\hsize-\width \FrameRestore}}%
    {\endMakeFramed}

\graphicspath{{graphics/}}

\newcommand{\edit}[1]{\textcolor{black}{#1}}

\DeclareMathOperator{\wt}{weight}

\nolinenumbers

\title{Universal cycle constructions for $k$-subsets and $k$-multisets}
\titlerunning{~~Universal cycles for $k$-subsets and $k$-multisets}

\author{Colin Campbell}{University of Guelph, Canada}{}{}{}
\author{Luke Janik-Jones}{University of Guelph, Canada}{}{}{}
\author{Joe Sawada}{University of Guelph, Canada}{}{}{}

\author{~}{~}{}{}{}
\authorrunning{~}

\Copyright{~}
\keywords{universal cycle, subsets, multisets, cycle-joining}

\begin{document}

%-----------  TITLE STUFF ---------------

%\title{ Decoding universal cycles for the $k$-subsets of an $n$-set}
%\author{ D. Gabric, W. Imam, L. Janik-Jones, J. Sawada \\
%}
\maketitle

\frenchspacing

\begin{abstract} 
A \emph{universal cycle} for a set $\bS$ of combinatorial objects is a cyclic sequence of length $|\bS|$ that contains a representation of each element in $\bS$ exactly once as a substring. 
If $\bS$ is the set of $k$-subsets of $[n] = \{\tt{1,2},\ldots, \tt{n}\}$, it is well-known that universal cycles do not always exists when applying a simple string representation, where $\tt{12}$ or $\tt{21}$ could represent the subset $\{\tt{1,2}\}$.  Similarly, if $\bS$ is the set of  $k$-multisets of $[n]$, it is also known that universal cycles do not always exist using a similar representation,  where $\tt{112}$, $\tt{121}$, or $\tt{211}$ could represent the multiset $\{\tt{1,1,2}\}$.  \edit{By mapping these sets to an appropriate family of labeled graphs, universal cycles are known to exist, but without a known efficient construction.  In this paper we consider a new representation for $k$-subsets and $k$-multisets that leads to efficient universal cycle constructions for all $n,k \geq 2$.}
%and demonstrate the existence of universal cycles for all $n,k \geq 2$.  Furthermore, 
We provide successor-rule algorithms to construct such universal cycles in $\mathcal{O}(n)$ time per symbol using $\mathcal{O}(n)$ space and demonstrate that necklace concatenation algorithms allow the same sequences to be generated in $\mathcal{O}(1)$ amortized time per symbol.   They are the first known efficient universal cycle constructions for $k$-multisets. 
The results are obtained by considering constructions for bounded-weight de Bruijn sequences. \edit{In particular, we demonstrate that a bounded-weight generalization of the Grandmama de Bruijn sequence can be constructed in $\mathcal{O}(1)$ amortized time per symbol.}
%By considering a shorthand frequency representation, we demonstrate the existence of the first known universal cycles for $k$-multisets.  Furthermore, we provide several interesting constructions that can be generated in $\mathcal{O}(n)$ time per symbol using $\mathcal{O}(n)$ space.  Furthermore, we demonstrate that each universal cycle can also be generated in $\mathcal{O}(1)$ amortized time per symbol.
\end{abstract}

%=====================================================================
%=====================================================================
%=====================================================================
\section{Introduction}  \label{sec:intro}

A \defo{universal cycle} for a set $\mathbf{S}$ of combinatorial objects is a cyclic sequence of length $|\mathbf{S}|$ that contains a \emph{representation} of each element in $\mathbf{S}$ exactly once as a substring. 
In this paper we focus on two sets: (a) $\bS_k(n)$, which denotes the set of $k$-subsets of $[n] = \{\tt{1,2},\ldots, \tt{n}\}$, and (b) $\bM_k(n)$, which denotes the set of $k$-multisets of $[n]$.  For example:
\begin{itemize}

\item $ \bS_2(3) = \Bigl\{ \{\tt{1,2}\} ,\{\tt{1,3}\}, \{\tt{2,3}\}  \Bigr\}$, and 

\item $ \bM_2(3) = \Bigl\{ \{\tt{1,1}\} ,\{\tt{1,2}\}, \{\tt{1,3}\},  \{\tt{2,2}\} ,\{\tt{2,3}\}, \{\tt{3,3}\} \Bigr\}. $
\end{itemize}
\noindent 
As we demonstrate in this paper, the choice of representation for a $k$-subset or $k$-multiset is critical to the existence \edit{and efficient construction} of universal cycles for these sets.

%When considering any algorithm for a given combinatorial object, the choice of representation can be critical to the existence, simplicity, and efficiency of any related algorithm.
%

\medskip

\noindent
{\bf Subsets.}~~Universal cycles for $\mathbf{S}_k(n)$ have been studied primarily by considering a \defo{standard} string representation, i.e., the subset $\{\tt{1,2}\}$ can be represented by either $\tt{12}$ or $\tt{21}$ (see~\cite{chung,2009research,lanius,rudoy}).  This representation has a major drawback as universal cycles exist only if $n$ divides $n \choose k$, or equivalently, if $k$ divides ${n-1 \choose k-1}$.  When this condition is met, Chung, Diaconis, and Graham~\cite{chung} conjectured that universal cycles for $\mathbf{S}_k(n)$ exist for sufficiently large $n$ once $k$ is fixed. The conjecture was verified for $k = 3,4,6$ in~\cite{jackson,hurlbert}, and universal cycles for subsets were shown not to exist when $k=n-2$~\cite{minustwo}. Recently, this conjecture was answered in the affirmative by Glock et al.~\cite{Glock} by studying Euler tours in hypergraphs.  The vast number of cases where universal cycles for $\mathbf{S}_k(n)$ do not exist using this standard representation has led to the study of subset packings and coverings~\cite{near-cycles,packings}.  
%To summarize, the choice of a standard representation for subsets has left no hope of a generic universal cycle construction, never mind efficient decoding routines.

%%%%%
Another common representation for a $k$-subset is a length-$n$ binary string with exactly $k$ ones.  It is straightforward to observe that non-trivial universal cycles do not exist using this representation.   
However, by considering a \defo{shorthand} binary string representation (one that omits the final redundant symbol), each $k$-subset corresponds to a unique length $n{-}1$ substring with either $k{-}1$ or $k$ ones.  Universal cycles for these strings, and hence $k$-subsets, can be constructed via a necklace concatenation approach~\cite{fixedweight2} in $\mathcal{O}(1)$ amortized time using $\mathcal{O}(n)$ space~\cite{fdgray13}.\footnote{The application of this universal cycle to subsets was first noted in~\cite{godbole2011}.} The same sequence can also be generated by an $\mathcal{O}(n)$ time per symbol successor rule~\cite{coolest} based on the ``missing symbol register''. %Unfortunately, there are no efficient decoding routines known for this construction.

In this paper, we consider a 
novel \defo{difference} representation that represents a $k$-subset with a length-$k$ string where the first symbol is the smallest element of the subset, and each successive symbol in the string is added to the previous symbol to obtain the next largest symbol in the subset. For example, the subset $\{\tt{1,3,4}\}$ for $n=5$ can be represented by $\tt{121}$.  The \defo{weight} of a string is the sum of its symbols. 
Observe that by applying the difference representation,
the set of $k$-subsets corresponds to all strings  over the alphabet $\{\tt{1,2,\ldots, n{-}k{+}1}\}$ of length $k$, where the weight of each string is bounded above by $n$. 
A similar notion of a difference representation was considered previously by Jackson~\cite{jackson} and Hurlbert~\cite{hurlbert}, however, it was considered cyclically and did not fix the first symbol.

%By mapping $\tt{x}$ to $\tt{x-1}$, the set of subsets can be represented by the set of all strings over the alphabet $\{0,1,\ldots , n{-}k\}$ of length $k$ with weight less than or equal to $n{-}k$.

%\begin{remark}
%A universal cycle for $\mathbf{S}_k(n)$ using the difference representation is a special case of universal cycles for $k$-ary strings of length~$n$ with an upper bound on the weight. 
%\end{remark}

%
The three different representations for $k$-subsets are illustrated in Table~\ref{tab:reps} for $\bS_3(5)$.
The sequence $\tt{1101011100}$ is a universal cycle for $\mathbf{S}_3(5)$ using the shorthand representation, and the sequence $\tt{1112122113}$
%$\tt{1113112212}$ 
%$\tt{3112212111}$
is a universal cycle for $\mathbf{S}_3(5)$ using the difference representation.  
There does not exist a universal cycle for $\mathbf{S}_3(5)$ using the standard representation. 

%
%================================
\begin{table}[h] 
    \centering
    \begin{tabular}{c|c|c|c}
        \multicolumn{1}{c|}{\bf Subset} & \multicolumn{1}{c|}{\bf Shorthand} & \multicolumn{1}{c|}{\bf Difference} & \multicolumn{1}{c}{\bf Standard}  \\
        \hline
        \{$\tt{1,2,3}$\}  & $\tt{1110~\textcolor{red}{(0)}}$ & $\tt{111}$ & $\tt{123, 132, \ldots, \mbox{ or } 321}$ \\
        \{$\tt{1,2,4}$\}  & $\tt{1101~\textcolor{red}{(0)}}$ & $\tt{112}$ & $\tt{124, 142, \ldots, \mbox{ or } 421}$\\
        \{$\tt{1,2,5}$\}  & $\tt{1100~\textcolor{red}{(1)}}$ & $\tt{113}$ & $\tt{125, 152, \ldots, \mbox{ or } 521}$\\
        \{$\tt{1,3,4}$\}  & $\tt{1011~\textcolor{red}{(0)}}$ & $\tt{121}$ & $\tt{134, 143, \ldots, \mbox{ or } 431}$\\
        \{$\tt{1,3,5}$\}  & $\tt{1010~\textcolor{red}{(1)}}$ & $\tt{122}$ & $\tt{135, 153, \ldots, \mbox{ or } 531}$\\
        \{$\tt{1,4,5}$\}  & $\tt{1001~\textcolor{red}{(1)}}$ & $\tt{131}$& $\tt{145, 154, \ldots, \mbox{ or } 541}$\\
        \{$\tt{2,3,4}$\}  & $\tt{0111~\textcolor{red}{(0)}}$ &  $\tt{211}$& $\tt{234, 243, \ldots, \mbox{ or } 432}$\\
        \{$\tt{2,3,5}$\}  & $\tt{0110~\textcolor{red}{(1)}}$ &  $\tt{212}$& $\tt{235, 253, \ldots, \mbox{ or } 532}$ \\
        \{$\tt{2,4,5}$\}  & $\tt{0101~\textcolor{red}{(1)}}$ &  $\tt{221}$& $\tt{245, 254, \ldots, \mbox{ or } 542}$ \\
        \{$\tt{3,4,5}$\}  & $\tt{0011~\textcolor{red}{(1)}}$ &  $\tt{311}$& $\tt{345, 354, \ldots, \mbox{ or } 543}$ \\
    
    \end{tabular}
    \caption{Illustrating different representations for the elements of $\mathbf{S}_3(5)$. For the shorthand representatives, the omitted redundant symbol is highlighted in parentheses.} 
    \label{tab:reps}
\end{table}
%================================

%\begin{exam}   \small
%The sequence $\tt{1101011100}$ is a universal cycle for $\mathbf{S}_3(5)$ using the shorthand representation, and the sequence $\tt{1113112212}$ 
%$\tt{3112212111}$
%is a universal cycle for $\mathbf{S}_3(5)$ using the difference representation.  
%There does not exist a universal cycle for $\mathbf{S}_3(5)$ using the standard representation.  
%\end{exam}

\edit{Cantwell et al.~\cite{cantwell} use a labeled graph to represent a $k$-subset.  Using this representation they demonstrate the existence of a universal cycle for a corresponding family of graphs; however, no efficient construction is provided.}

\medskip

\noindent
{\bf Multisets.}~~Like with $k$-subsets, the choice of representation is critical to constructing universal cycles for $k$-multisets.  It is well known that $|\bM_k(n)| = {n + k -1 \choose k}$.  Universal cycles for $\bM_k(n)$ were first considered by Jackson~\cite{jackson} and Hurlbert, Johnson, and Zahl~\cite{hurlbert2009} using a \defo{standard} string representation, where $\tt{112}$, $\tt{121}$, or $\tt{211}$ could represent the multiset $\{\tt{1},\tt{1},\tt{2}\}$.    Similar to $k$-subsets, they  demonstrate that universal cycles for $\bM_k(n)$ exists only if $n$ divides ${n+k-1 \choose k}$.  Another method to represent a $k$-multiset is with a \defo{frequency map}, which is a length-$n$ string $f_1f_2\cdots f_n$ where each $f_i$ represents the number of occurrences of $i$ in the multiset. Note, $\sum_{i=1}^n f_i = k$.  
For instance, $110$ is the frequency map representation for the multiset $\{1,2\}$ in $\bM_2(3)$.  
%Using this representation, the multisets of $\bM_2(3)$ noted earlier can be represented uniquely by $\{ \tt{200, 110, 101, 020, 011, 002} \}$. 
It is a simple exercise to demonstrate that universal cycles for $\bM_k(n)$ do not exist using this representation for $n,k \geq 2$. %For instance, any universal cycle for $\bM_3(3)$ must contain 002 and to maintain the frequency, the three symbols following this string must also be 002. 
However, observe that the final symbol $f_n$ in a frequency map is redundant: its value can be determined from the previous $n-1$ values,  $f_n = k -  \sum_{i=1}^{n-1} f_i$.  Thus, we say $f_1\cdots f_{n-1}$ is a \defo{shorthand frequency} representation for a multiset in $\bM_k(n)$.  
%Using this representation, the multisets of $\bM_2(3)$ can be represented by $\{ \tt{20, 11, 10, 02, 01, 00} \}$.
Observe that this set corresponds to all strings over the alphabet $\{\tt{0,1},\ldots, \tt{k}\}$ of length $n-1$,
where the weight of each string is bounded above by $k$.  
%A universal cycle for this set is $\tt{001102}$. 

A \defo{difference} representation can also be used for $k$-multisets.  Consider the same definition applied for $k$-subsets except assign the first symbol of the difference representative to be \emph{one less} than the smallest symbol in the multiset.  Equivalently, apply the original definition of a difference representative directly, but define the $k$-multisets to be over the ground set $\{0,1, \ldots, n-1\}$.  Observe that by applying this difference representation,
the set of $k$-multisets corresponds to all strings  over the alphabet $\{\tt{0,1,\ldots, n{-}1}\}$ of length $k$, where the weight of each string is bounded above by $n-1$. 

These three different representations for $k$-multisets are illustrated in Table~\ref{tab:reps2} for $\bM_3(3)$.
The sequence $\tt{0011021203}$ is a universal cycle for $\mathbf{M}_3(3)$ using the shorthand frequency representation, and the sequence $\tt{0001011002}$
%$\tt{1113112212}$ 
%$\tt{3112212111}$
is a universal cycle for $\mathbf{M}_3(3)$ using the difference representation.  
There does not exist a universal cycle for $\mathbf{M}_3(3)$ using the standard representation.

%================================
\begin{table}[h] 
    \centering
    \begin{tabular}{c|c|c|c}
        \multicolumn{1}{c|}{\bf Multiset} & \multicolumn{1}{c|}{\bf Shorthand frequency} & \multicolumn{1}{c|}{\bf Difference} & \multicolumn{1}{c}{\bf Standard}  \\
        \hline
        \{$\tt{1,1,1}$\}  & $\tt{30~\textcolor{red}{(0)}}$ & $\tt{000}$ & $\tt{111}$ \\
        \{$\tt{1,1,2}$\}  & $\tt{21~\textcolor{red}{(0)}}$ & $\tt{001}$ & $\tt{112,121,\mathrm{or~}211}$\\
        \{$\tt{1,1,3}$\}  & $\tt{20~\textcolor{red}{(1)}}$ & $\tt{002}$ & $\tt{113,131,\mathrm{or~}311}$\\
        \{$\tt{1,2,2}$\}  & $\tt{12~\textcolor{red}{(0)}}$ & $\tt{010}$ & $\tt{122,212,\mathrm{or~}221}$\\
        \{$\tt{1,2,3}$\}  & $\tt{11~\textcolor{red}{(1)}}$ & $\tt{011}$ & $\tt{123,132,213,231,312,\mathrm{or~}321}$\\
        \{$\tt{1,3,3}$\}  & $\tt{10~\textcolor{red}{(2)}}$ & $\tt{020}$& $\tt{133,313,\mathrm{or~}331}$\\
        \{$\tt{2,2,2}$\}  & $\tt{03~\textcolor{red}{(0)}}$ &  $\tt{100}$& $\tt{222}$\\
        \{$\tt{2,2,3}$\}  & $\tt{02~\textcolor{red}{(1)}}$ &  $\tt{101}$& $\tt{223,232,\mathrm{or~}322}$ \\
        \{$\tt{2,3,3}$\}  & $\tt{01~\textcolor{red}{(2)}}$ &  $\tt{110}$& $\tt{233,323,\mathrm{or~}332}$ \\
        \{$\tt{3,3,3}$\}  & $\tt{00~\textcolor{red}{(3)}}$ &  $\tt{200}$& $\tt{333}$ \\
    
    \end{tabular}
    \caption{Illustrating different representations for the elements of $\mathbf{M}_3(3)$. For the shorthand frequency representatives, the omitted redundant symbol is highlighted in parentheses.} 
    \label{tab:reps2}
\end{table}
%================================

 %An application of universal cycles for $k$-multisets to proximity sensor networks is discussed in~\cite{chen2024}.  A multiset $M$ is said to be $t$-restricted if the number of times that each element appears in $M$ is at most $t$. Instances for the existence of universal cycles for $t$-restricted multisets is considered in~\cite{godbole2011}.

Universal cycles for $k$-multisets have applications related to proximity sensor networks~\cite{chen2024}.
If the number of each element in a $k$-multiset is bounded above by $t$, then such $t$-restricted multisets have been studied in~\cite{godbole2011}. 
\edit{Like with $k$-subsets, Cantwell et al.~\cite{cantwell} use a labeled graph to represent a $k$-multiset.  Using this representation they demonstrate the existence of a universal cycle for a corresponding family of graphs; however, no efficient construction is provided.}

%
%\red{re-write this part}
%Eight $\mathcal{O}(n)$-time successor-rule constructions are provided in~\cite{karyframework}; four bound the weight from below, and four bound the weight from above.  Concatenation approaches for each construction are also given in~\cite{concattree}. %, where each symbol can be generated in $\mathcal{O}(1)$ amortized time. Is this TRUE??
%De Bruijn sequence constructions  with both lower and upper bounds on the weight  are considered in~\cite{sawada2013}.  None of these constructions were previously known to be efficiently decoded.  %In this paper demonstrate that one can be efficiently decoded and leverage the result to efficiently decode universal cycles for $k$-subsets.
%

%%%%

\bigskip

%=====================================================================
\noindent
{\bf Main Results.} ~Recall that the
difference representatives of $\bS_k(n)$ and both the shorthand frequency and difference representatives of $\bM_k(n)$ are each special cases of fixed-length strings over a prescribed alphabet, where the strings have an upper bound on their weight.  Thus,   results for bounded-weight de Bruijn sequences (see Section~\ref{sec:bound}) can be applied to obtain the main results of this paper. \edit{ In particular, we demonstrate that a bounded-weight generalization of the Grandmama de Bruijn sequence~\cite{grandma2} can be constructed in $\mathcal{O}(1)$ amortized time per symbol. We also present a novel construction by applying the missing symbol register (MSR).}
\begin{enumerate}

    \item We demonstrate an algorithm to construct a universal cycle $\mathcal{S}_1$ for $\bS_k(n)$ using  difference representatives that runs  in $\mathcal{O}(1)$ amortized time per symbol.  Moreover, given any subset in $\bS_k(n)$ with difference representative $d_1d_2\cdots d_k$, the symbol following this string in $\mathcal{S}_1$ can be computed in $\mathcal{O}(n)$ time. 

\smallskip

    \item We demonstrate an algorithm to construct a universal cycle $\mathcal{M}_1$ for $\bM_k(n)$ using  shorthand frequency representatives that runs in $\mathcal{O}(1)$ amortized time per symbol.  Moreover, given any subset in $\bM_k(n)$ with shorthand frequency representative $f_1f_2\cdots f_{n-1}$, the symbol following this string in $\mathcal{M}_1$ can be computed in $\mathcal{O}(n)$ time.  

    \item We demonstrate an algorithm to construct a universal cycle $\mathcal{M}_2$ for $\bM_k(n)$ using difference representatives that runs in $\mathcal{O}(1)$ amortized time per symbol.  Moreover, given any subset in $\bM_k(n)$ with difference representative $d_1d_2\cdots d_{k}$, the symbol following this string in $\mathcal{M}_2$ can be computed in $\mathcal{O}(n)$ time. 
\end{enumerate}
\noindent
The last two results are the first known efficient universal cycle constructions for $k$-multisets. 
All the algorithms require polynomial space with respect to $n$ and $k$.
 Implementations for many known universal cycle constructions, including the algorithms presented in this paper, are available at \url{http://debruijnsequence.org}~\cite{dbseq}.

\bigskip

\noindent
{\bf Outline of paper.}
In Section~\ref{sec:back} we present preliminary definitions, notation, and provide a brief background on the cycle-joining process and concatenation trees.  In Section~\ref{sec:bound} we describe a known algorithm to construct universal cycles for $t$-ary strings of length $n$ with a given weight constraint and provide new insights into one particular construction. We also introduce a new construction with interesting properties.  Then in Section~\ref{sec:results}, we apply these constructions to $k$-subsets and $k$-multisets.  
%We conclude in Section~\ref{sec:fut} with avenues for future research.

%=====================================================================
%=====================================================================
%=====================================================================
\section{Preliminaries} \label{sec:back}

Let $\Sigma_t = \{\tt{0,1,2},\ldots, \tt{t{-}1}\}$ and let $\Sigma_t(n)$ denote the set of all length-$n$ strings over $\Sigma_t$.  Recall, the \defo{weight} of a string is the sum of its symbols. Let $\wt(\alpha)$ denote the weight of a string $\alpha$. 
Let $\Sigma_t(n,w)$ denote the subset of all strings in $\Sigma_t(n)$ with weight \emph{at most} $w$.   Later in this paper we will be interested in strings as they are listed in \defo{colex} order, which is standard lexicographic order when the strings are read from right to left. For instance, the following set of bounded-weight strings is listed in colex order:
\[ \Sigma_3(3,2) = \{\tt{000, 100, 200, 010, 110, 020, 001, 101, 011, 002}\}.\]
Consider two strings $\alpha$ and $\beta$.
Let  $\alpha \cdot \beta$ denote the  concatenation of $\alpha$ and $\beta$, and   let $\beta^j$ denote $j$ copies of $\beta$ concatenated together.
The \defo{aperiodic prefix} of $\alpha$ is the shortest string $\beta$ such that $\alpha = \beta^j$ for some $j \geq 1$.
A string is said to be \defo{aperiodic} if it is the same as its aperiodic prefix; otherwise, it is said to be \defo{periodic}.

A \defo{necklace class} is an equivalence class of strings under rotation.
A \defo{necklace} is the lexicographically smallest string in a necklace class.   Let $\bN_t(n)$ denote the set of length-$n$ necklaces over $\Sigma_t$.  For example, 
\[ \bN_3(3) = \{\tt{000, 001, 002, 011, 012, 021, 022, 111,112,122, 222} \}. \]
Let $\bN_t(n,w)$ denote the subset of all strings in $\bN_t(n)$ with weight at most $w$.  For example, 
\[ \bN_3(3,2) = \{\tt{000, 001, 002, 011} \}. \]
%
  
%Let $\ap(\alpha)$ denote the \defo{aperiodic prefix} of $\alpha$, i.e., the shortest string $\beta$ such that $\alpha = \beta^j$ for some $j \geq 1$.
  % For example, $\ap(\tt{111}) = \tt{1}$.  
  % A string $\alpha$ is said to be \defo{aperiodic} (primitive) if $\ap(\alpha) = \alpha$; otherwise $\alpha$ is \defo{periodic} (non-primitive).  

A \defo{feedback function} $f$ maps strings from $\Sigma_t(n)$ to $\Sigma_t$.
A \defo{feedback shift register} is a function that maps a string $\alpha = \tt{a}_1\tt{a}_2 \cdots \tt{a}_n$ to $\tt{a}_2 \cdots \tt{a}_n f(\alpha)$ for a given feedback function $f$.  
The \defo{pure cycling register} (PCR) is a feedback shift register with feedback function $f(\tt{a}_1\tt{a}_2\cdots \tt{a}_n) = \tt{a}_1$.  It partitions any set closed under rotation into necklace classes that can be represented by necklaces from $\bN_t(n)$.   Consider a subset $\bS$ of strings in $\Sigma_t(n+1)$ where each string has a unique shorthand representative like those outlined in Section~\ref{sec:intro}.  Let $\bS'$ denote this set of shorthand representatives. Then each string
$\alpha = \tt{a}_1\tt{a}_2\cdots \tt{a}_n \in \bS'$ corresponds to a unique string $\tt{a}_1\tt{a}_2\cdots \tt{a}_n\tt{z} \in \bS$.  We call $\tt{z}$ the missing symbol. The 
\defo{missing symbol register} (MSR) is a feedback shift register for such sets $\bS'$ with feedback function $f(\tt{a}_1\tt{a}_2\cdots \tt{a}_n) = \tt{z}$.  The MSR partitions any shorthand set $\bS'$ closed under rotation into necklace classes that can be represented by necklaces in $\bN_t(n+1)$ (see Example~\ref{ex:msr} in Section~\ref{sec:msr}). The MSR has been applied to shorthand representatives of permutations, subsets, and more generally, strings with fixed content in~\cite{coolest}.

%==========================================================================
%==========================================================================
%==========================================================================
\subsection{Cycle-joining trees and successor rules} \label{sec:cycle-join}

In this section, assume that the representatives for a set of strings $\bS$ is just $\bS$ itself.
Given a universal cycle $\mathcal{U} = \tt{u}_1\tt{u}_2 \cdots \tt{u}_m$ for  $\mathbf{S}$, a 
\defo{successor rule} for $\mathcal{U}$ is a function that maps each string $\alpha$ in $\mathbf{S}$ to the symbol following $\alpha$ in $\mathcal{U}$.  Thus, starting with any string in $\mathbf{S}$,  $\mathcal{U}$ (considered cyclically) can be constructed by repeatedly applying its successor rule.

Let $\mathcal{U}_i$ denote a universal cycle for $\mathbf{S}_i \subseteq \Sigma_t(n)$. Two universal cycles $\mathcal{U}_1$ and $\mathcal{U}_2$ are said to be \defo{disjoint} if $\mathbf{S}_1 \cap \mathbf{S}_2 = \emptyset$.
Let $\tt{x},\tt{y}$ be distinct symbols in $\Sigma_t$.
If $\sigma = \tt{x}\tt{s}_2\cdots \tt{s}_n$ and $\hat \sigma = \tt{y}\tt{s}_2\cdots \tt{s}_n$, then $\sigma$ and $\hat \sigma$ are said to be \defo{conjugates} of each other, and $(\sigma, \hat \sigma)$ is called a \defo{conjugate pair}.   Two disjoint universal cycles $\mathcal{U}_1$ and $\mathcal{U}_2$ can be joined together to form a single universal cycle by swapping the successor of each string in a conjugate pair $(\sigma, \hat \sigma)$, where $\sigma$ is found in $\mathcal{U}_1$ and $\hat \sigma$ is found in $\mathcal{U}_2$.  This is the well-known \defo{cycle-joining process} where the two cycles $\mathcal{U}_1$ and $\mathcal{U}_2$ are joined via the conjugate pair $(\sigma, \hat \sigma)$.
A \defo{cycle-joining tree} $\mathbb{T}$  is an \emph{unordered} tree where 
the nodes correspond to a disjoint set of universal cycles $\mathcal{U}_1, \mathcal{U}_2, \ldots, \mathcal{U}_\ell$ and an edge between $\mathcal{U}_i$ and $\mathcal{U}_j$ is defined by a conjugate pair $(\sigma, \hat \sigma)$ such that $\sigma \in \mathbf{S}_i$ and $\hat \sigma \in \mathbf{S}_j$. 
For our purposes, we consider cycle-joining trees to be rooted.  
In the binary case, repeatedly joining adjacent cycles in a cycle-joining tree yields a \emph{unique} universal cycle for $\bS_1 \cup \bS_2 \cup \cdots \cup \bS_\ell$.  
For non-binary alphabets, different universal cycles are possible depending on the order that the cycles are joined~\cite{concattree}.  Thus, the following property was introduced to produce a generic successor rule for non-binary alphabets based on the conjugate pairs.

%When two adjacent nodes $\mathcal{U}_i$ and $\mathcal{U}_j$ in a cycle-joining tree $\mathbb{T}$ are joined  to obtain $\mathcal{U}$ via Lemma~\ref{thm:concat} (rotating the cycles as appropriate), the nodes are unified and replaced with $\mathcal{U}$ (the edge between $\mathcal{U}_i$ and $\mathcal{U}_j$ is contracted).  Repeating this process until only one node remains produces a universal cycle for $\mathbf{S}_1 \cup \mathbf{S}_2 \cup \cdots \cup \mathbf{S}_t$.  In the binary case, the same universal cycle is produced, no matter the order in which the cycles are joined.  This is because no string can belong to more than one conjugate pair in the underlying definition of $\mathbb{T}$.  However, when $k>2$, the order that the cycles are joined can be important.  

%In upcoming discussion regarding both successor rules and concatenation trees, we require the  underlying cycle-joining trees to have the following property when $k>2$.

\begin{result}
\noindent
{\bf Chain Property}: If a node in a cycle-joining tree $\mathbb{T}$ has two children joined via conjugate pairs $(\tt{x}\tt{s}_2\cdots \tt{s}_n, \tt{y}\tt{s}_2\cdots \tt{s}_n)$ and 
$(\tt{x}'\tt{t}_2\cdots \tt{t}_n, \tt{y'}\tt{t}_2\cdots \tt{t}_n)$, then $\tt{s}_2\cdots \tt{s}_n \neq \tt{t}_2\cdots \tt{t}_n$. 
\end{result}

Let $\mathbb{T}$  be a cycle-joining tree  satisfying the Chain Property for an underlying set of strings $\bS$, where the nodes are joined by a set $\mathbf{C}$ of conjugate pairs. If the tree contains $\ell$ cycles then $\mathbf{C}$ contains $\ell{-}1$ conjugate pairs, each corresponding to an edge in $\mathbb{T}$.  Say $\gamma$ \defo{belongs to} a conjugate pair $(\sigma, \hat \sigma)$ if either $\gamma = \sigma$ or $\gamma = \hat \sigma$.  
Let $C_1, C_2, \ldots, C_m$ denote a maximal length path of nodes in $\mathbb{T}$ such that for each $1 \leq i < m$, the node $C_i$ is the parent of $C_{i+1}$
and they are joined via a conjugate pair of the form  $(\tt{x}_i\beta, \tt{x}_{i+1}\beta$); $\beta$ is the same in each conjugate pair.  
Call such a path a \defo{chain} of length $m$, and define $g(\tt{x}_i\beta) = \tt{x}_{i+1}$ for $i < m$ and $g(\tt{x}_m\beta) = \tt{x}_1$ for each string belonging to a conjugate pair.
%
%Let $\alpha = \tt{a}_1\tt{a}_2\cdots \tt{a}_n$ be a string in $\bS$.  
%If $\alpha = \tt{x}_i\beta$ belongs to a conjugate pair that joins two nodes in such a chain of $m$ nodes
%$\tt{x}_1\beta, \tt{x}_2\beta, \ldots, \tt{x}_m\beta$, 
%let $g(\alpha) = \tt{x}_{i+1}$ if $i < m$ and $g(\alpha) = \tt{x}_1$ if $i = m$. 
%Let $\alpha = \tt{x}_i\beta$ be a string in $\bS$ that belongs to a conjugate pair that joins two nodes in such a chain.
%Define $g(\alpha) = \tt{x}_{i+1}$ if $i < m$ and $g(\alpha) = \tt{x}_1$ if $i = m$.
%
Then the following function $h$ is a successor rule for a corresponding universal cycle for $\mathbf{S}$,  where $f(\alpha)$ is the underlying feedback function (such as the PCR or MSR) that induces the cycles in $\mathbb{T}$:

\begin{center}
$h(\alpha) = \left\{ \begin{array}{ll}

        g(\alpha)  &\ \  \mbox{if $\alpha$ belongs to a conjugate pair  in $\mathbf{C}$;}\\
         f(\alpha) \  &\ \  \mbox{otherwise.}\end{array} \right. $
         
\end{center}
This rule corresponds to the function ${\uparrow}f_1(\alpha)$ in \cite{concattree} when $f$ is the feedback function for the PCR.  

\begin{exam}
Consider the cycle-joining tree with nodes induced by the PCR illustrated in Figure~\ref{fig:tree}.
The path of nodes $\tt{000} \rightarrow \tt{001}\rightarrow \tt{002} \rightarrow \tt{003} \rightarrow \tt{004}$ form a chain where $\beta = 00$. The conjugate pairs are
$(\tt{000},\tt{100})$,  $(\tt{100},\tt{200})$, $(\tt{200},\tt{300})$, and $(\tt{300},\tt{400})$, respectively.  Thus, 
$h(\tt{000}) = \tt{1}$, 
$h(\tt{100}) = \tt{2}$, 
$h(\tt{2000}) = \tt{3}$, 
$h(\tt{300}) = \tt{4}$, and  
$h(\tt{400}) = \tt{0}$. 
\end{exam}

\noindent
 Applying the rule directly requires storing all the conjugate pairs; however, in our application it is relatively straightforward to test whether a string belongs to a conjugate pair in $\mathcal{O}(n)$ time using only $\mathcal{O}(n)$ space.
In this paper, we will define cycle-joining trees based on the PCR and MSR and then apply $h(\alpha)$ to construct corresponding universal cycles efficiently. 

%=====================================================================
%=====================================================================
%=====================================================================
\subsection{Concatenation trees} \label{sec:concattree}

In this section we present a simplified presentation of ``left'' concatenation trees as originally defined in~\cite{concattree}, by adding an assumption that every non-root periodic node in a cycle-joining tree is a leaf.

A \defo{bifurcated ordered tree} (BOT) is a rooted tree where each child of a node belongs to either the \defo{left-children} or the \defo{right-children}, and the nodes within each group are ordered. 
Let $\mathbb{T}$ be a PCR-based cycle-joining tree rooted at $r$ satisfying the Chain Property. 
A \defo{concatenation tree} %\footnote{We define a ``left'' concatenation tree in our presentation: a node with the same change index as its parent is defined as a leftchild.} 
based on $\mathbb{T}$ converts $\mathbb{T}$ into a BOT
by ordering the children, splitting them into left-children and right-children, and assigning possibly new labels to represent each node (cycle).  The parent-child relationship remains the same and each node has a \defo{change index}, which is the unique index where a node's label differs from that of its parent based on the conjugate pair joining the two nodes.  This index is unique using our added assumption that all non-root periodic nodes are leaves, except for the case which occurs if a node is the child of a periodic root. We will handle that case as it arises.  The root $r$ can be assigned an arbitrary change index.   If a node has change index $c$, all children with a change index less than \emph{or equal to $c$} are classified as left children and are ordered from smallest to largest change index.  All children with a change index greater than $c$ are classified as right children and are also ordered from smallest to largest change index. An \defo{RCL ordering} traverses a concatenation tree recursively from the root by first visiting all {\bf R}ight children, then the {\bf C}urrent node, then the {\bf L}eft children.  For example, see Figure~\ref{fig:tree} which illustrates a cycle-joining tree, a corresponding concatenation tree, and its RCL ordering.  

Given a concatenation tree $T$, let $\mathcal{U}_T$ be the sequence obtained by concatenating the aperiodic prefixes of each node as they are visited in RCL order. The following result corresponds to Theorem 3 from~\cite{concattree}.

\begin{theorem} \label{thm:concat}
Let $T$ be a concatenation tree for a PCR-based cycle-joining tree with the Chain Property for an underlying set $\bS$.  
Then $\mathcal{U}_T$ is a universal cycle for $\bS$ with successor rule $h(\alpha)$.
\end{theorem}

%=====================================================================
%=====================================================================
%=====================================================================
\section{Bounded-weight de Bruijn sequences} \label{sec:bound}

A universal cycle for $\Sigma_t(n)$ is known as a \defo{de Bruijn sequence}.  A universal cycle for the subsets of $\Sigma_t(n)$ where there is an upper or lower bound on the weight is called a \defo{bounded-weight de Bruijn sequence}.  The lexicographically smallest binary de Bruijn sequence with a lower bound on the weight can be constructed in $\mathcal{O}(1)$ amortized time per symbol~\cite{weight1}.  The algorithm is generalized to non-binary alphabets in~\cite{generalize-classic-greedy} by providing an equivalent greedy construction.  
There are eight $\mathcal{O}(n)$ time successor-rule constructions for bounded-weight de Bruijn sequences over an arbitrary alphabet given in~\cite{karyframework}: four with a lower bound on the weight and four with an upper bound on the weight.\footnote{Na\"{i}vely, each successor rule requires up to $t$ necklace tests, each requiring $\mathcal{O}(n)$ time.  However applying optimizations similar to the upcoming proof of Theorem~\ref{thm:MSR}, these successor rules can be implemented with a constant number of necklace tests.} 
Each of the above algorithms requires $\mathcal{O}(n)$ space. 
We are interested in universal cycles for $\Sigma_t(n,w)$, which are strings with an upper bound of $w$ on the weight.

In this section we focus on one of the successor rules from~\cite{karyframework} that constructs a universal cycle for $\Sigma_t(n,w)$ with interesting properties. 
Since $\Sigma_t(n,w)$ is closed under rotation, the PCR partitions the set into necklace cycles. These cycles (nodes) can be represented by the necklaces in $\bN_t(n,w)$.  
For example, the cycles for $\Sigma_5(3,4)$ can be represented by the following necklaces as they are listed in colex order (this order will be useful later):
%
%\begin{equation} \label{eq:colex}
 \[   \bN_5(3,4) =  \{ \tt{000, 001, 011, 111, 021, 031, 002, 012, 112, 022, 003, 013, 004} \}. \]
%\end{equation}
%
The following parent rule\footnote{This rule was called \emph{first non-1} in~\cite{karyframework} since the alphabet considered was $\{1,2, \ldots ,t\}$.} defines a cycle-joining tree rooted at $\tt{0}^n$ for the  cycles represented by the necklaces in $\bN_t(n,w)$. 

%
%=================
\begin{result} \noindent
{\bf First non-zero parent rule (with root $0^{n}$)}: the parent of non-root node $\tt{a}_1\tt{a}_2\cdots \tt{a}_{n}$, where $j$ denotes the smallest index such that $\tt{a}_j \neq 0$, is 
%$\tt{a}_1\cdots \tt{a}_{j-1}(\tt{a}_j{-}1)\tt{a}_{j+1}\cdots \tt{a}_{n} = \tt{0}^{j-1}(\tt{a}_j{-}1)\tt{a}_{j+1}\cdots \tt{a}_{n}$.  
$\tt{0}^{j-1}(\tt{a}_j{-}1)\tt{a}_{j+1}\cdots \tt{a}_{n}$.  
\end{result}
%=================
%
\noindent
As an example, the cycle-joining tree for $\Sigma_5(3,4)$  based on this \emph{first non-zero} parent rule is shown in Figure~\ref{fig:tree}.
%
%=================
\begin{figure} [h]
\begin{center}
  \includegraphics[scale=0.95]{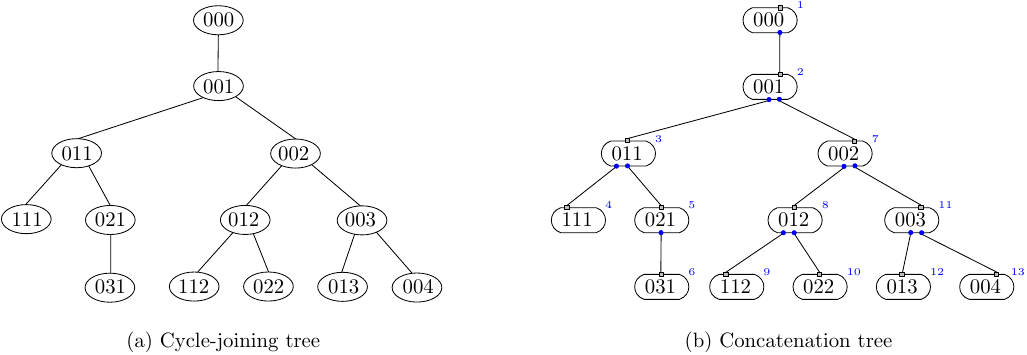}
  \caption{ (a) The first non-zero cycle-joining tree for $\Sigma_5(3,4)$ where the nodes are represented by the necklaces in $\bN_5(3,4)$. (b) A corresponding concatenation tree where the change index of the root is assigned to the final symbol, as indicated by the small box on top of the node. All children are left children, and the labels on the nodes do not need to change.  An RCL traversal of this tree visits the necklaces in colex order. }
  \label{fig:tree}
\end{center}
\end{figure}
%=================
%
%In the binary case, a cycle-joining tree yields a unique successor rule for the corresponding universal cycle.  However, for non-binary alphabets, different universal cycles are possible depending on the order that the cycles are joined (see an example in~\cite{concattree}). 
%

%The following results regarding the cycle-joining tree induced by the  first non-zero parent rule will be applied when we prove the main result of this section.
%

\begin{lemma} \label{lem:props}
Let $\mathbb{T}$ be a PCR-based cycle-joining tree induced by the  first non-zero parent rule and let $\alpha = \tt{0}^{i-1}\tt{x}\tt{a}_{i+1}\cdots \tt{a}_n$ and
$\beta = \tt{0}^{j-1}\tt{y}\tt{b}_{j+1}\cdots \tt{b}_n$ be two distinct nodes (necklaces) in $\mathbb{T}$ where $i,j \leq n$ and $\tt{x},\tt{y} > 0$.  Then the following hold. 
\begin{enumerate}

    \item If $\alpha \neq \tt{0}^{n-1}\tt{1}$ then the parent of $\alpha$ is an aperiodic necklace. (All non-root periodic nodes are leaves.)
        \item $\mathbb{T}$ has the Chain Property.
    \item If $\alpha$ and $\beta$ are periodic then they have different parents.
\end{enumerate}
\end{lemma}
\begin{proof}
{\bf (1)} Since $\alpha$ is a necklace, its prefix $\tt{a}_1\cdots \tt{a}_i = \tt{0}^{i-1}\tt{a}_i$ is lexicographically less than or equal to every other substring in $\alpha$ of the same length when $\alpha$ is considered cyclically. Its parent $\gamma$ has prefix $\tt{0}^{i-1}(\tt{a}_i-1)$
which is strictly less than all other substrings in $\gamma$ considered cyclically, which means $\gamma$ is a necklace. Furthermore, this also implies $\gamma$ is aperiodic as long as $\gamma \neq \tt{0}^n$, which is the case only when $\alpha = \tt{0}^{n-1}1$.

{\bf (2-3)}
Suppose $\alpha$ and $\beta$ have the same parent $\gamma$.  If $i=j$ then $\alpha = \beta$, a contradiction to $\alpha$ and $\beta$ being distinct. Since they are necklaces where $\tt{x},\tt{y} > 0$, we have $\tt{a}_n, \tt{b}_n > \tt{0}$. Without loss of generality, suppose $i < j$.  The strings in the conjugate pair joining $\alpha$ and $\gamma$ have suffix $\tt{a}_n\tt{0}^{i-1}$ and the strings in the conjugate pair joining $\beta$ and $\gamma$ have suffix $\tt{b}_n\tt{0}^{j-1}$.  Thus, the Chain Property is satisfied. Furthermore, suppose $\alpha$ and $\beta$ are both periodic. Since $\beta$ is periodic, it must be that $\gamma$ has a substring $0^{j-1}$ appearing after the $j$-th symbol.  But this implies that $\alpha$ also has the same substring since  $\alpha$ and $\beta$ share the same parent. This contradicts the fact that $\alpha$ is a necklace.
Thus $\alpha$ and $\beta$ have different parents.  
\end{proof}

The following successor rule based on the first non-zero cycle-joining tree corresponds to the successor rule $g_4$ in~\cite{karyframework}. In  its original presentation, the alphabet under consideration is $\{\tt{1,2, \ldots ,t}\}$ whereas we state the rule for the alphabet $\Sigma_t = \{\tt{0,1, \ldots ,t{-}1}\}$.
It is also a space-efficient implementation of the generic successor rule $h(\alpha)$ presented in Section~\ref{sec:cycle-join} that does not require storing the set of conjugate pairs.

%========================================================
% GRANDMAMA SUCCESSOR RULE
%========================================================
\begin{result} 
\noindent {\bf First non-zero (Grandmama)  successor rule} for $\alpha = \tt{a}_1\tt{a}_2\cdots \tt{a}_{n}$ 

\smallskip

\noindent
Let $j$ be the largest index of $\tt{a}_2\tt{a}_3\cdots \tt{a}_{n}$ such that $\tt{a}_j \ne \tt{0}$ or $j=1$ if no such index exists.
Let $\tt{x}$ be the largest symbol in $\{\tt{1,2, \ldots, t{-}1}\}$ such that $\tt{0}^{{n}-j}   \tt{x}  \tt{a}_2\cdots \tt{a}_{j}$ is a necklace and $\wt(\tt{x}\tt{a}_2\tt{a}_3\cdots \tt{a}_{{n}}) \leq w$,
or let $\tt{x}=-1$ if no such symbol exists.

\begin{center}  \small
$h_1(\alpha) = \left\{ \begin{array}{ll}
          0 		&\ \   \mbox{if  $\tt{x} \ne -1$ and $\tt{a}_1 = \tt{x}$;}\\
         \tt{a}_1{+}1	&\ \  \mbox{if  $\tt{x} \ne -1$ and $\tt{a}_1 < \tt{x}$;}   \ \ \ \ \ \ \ \ \\
         \tt{a}_1 \  	&\ \  \mbox{otherwise.}\end{array} \right.$
\end{center}

\vspace{-0.1in}
\end{result}
%========================================================
%
\noindent
Let $\mathcal{U}_t(n,w)$ denote the universal cycle for $\Sigma_t(n,w)$ obtained by starting with $\tt{0}^n$ and repeatedly applying the successor rule $|\Sigma_t(n,w)|-n$ times on the last $n$ symbols to obtain the next symbol in the sequence.  For example:
\[ \mathcal{U}_5(3,4) = \tt{\underline{0 \cdot 00}1 \cdot 011 \cdot 1 \cdot 021 \cdot 031 \cdot 002 \cdot 012 \cdot 112 \cdot 022 \cdot 003 \cdot 013 \cdot 004.} \]
Observe that this sequence has an interesting property: it corresponds to concatenating the aperiodic prefixes of the necklaces in $\bN_5(3,4)$ as they appear in colex order.  
This is not surprising since when there is no bound on the weight, i.e., $w \geq n(t{-}1)$, the first non-zero successor rule is equivalent to the successor rule for the \defo{Grandmama de Bruijn sequence}~\cite{grandma2}.  The Grandmama sequence can also be described by the following very simple concatenation scheme:
\begin{quote}
\emph{Concatenate together the aperiodic prefixes of the necklaces in $\bN_t(n)$ as they appear in colex order.} 
\end{quote}
Applying Theorem~\ref{thm:concat}, this concatenation construction can be generalized to the set $\Sigma_t(n,w)$.  
Let $\mathbb{T}$ denote the cycle-joining tree for $\Sigma_t(n,w)$ defined by the first non-zero parent rule, where the cycles (nodes) are induced by the PCR -- they correspond to the necklaces in $\bN_t(n,w)$.  
%The first non-zero  successor rule $h_1(\alpha)$ used to construct $\mathcal{U}_t(n,w)$ is a space efficient implementation of $h(\alpha)$. 
From Lemma~\ref{lem:props}, $\mathbb{T}$ satisfies the Chain Property and all the non-root periodic nodes are leaves.  Thus, let $T$ be the concatenation tree derived from $\mathbb{T}$, where the change index of the root node $\tt{0}^n$ is the index of the rightmost symbol. Define the label of the single child of the root to be $\tt{0}^{n-1}\tt{1}$. This accounts for the one ambiguity in our simplified concatenation tree definition and this choice abides by the conditions from the original definition in~\cite{concattree}.  Recall that the labels of all other nodes are defined recursively based on the label of their parent.  It follows from Lemma~\ref{lem:props}, that the labels of all other nodes are necklaces.

\begin{theorem} \label{thm:main}
For $n,t \geq 1$ and $w \geq 0$, the universal cycle $\mathcal{U}_t(n,w)$ for $\Sigma_t(n,w)$ can be constructed by concatenating the aperiodic prefixes of the necklaces in $\bN_t(n,w)$ as they appear in colex order. 
\end{theorem}
\begin{proof}
By Theorem~\ref{thm:concat}, the sequence obtained by concatenating the aperiodic prefixes of the nodes in $T$ as they appear RCL order has successor rule $h_1(\alpha)$, which generates $\mathcal{U}_t(n,w)$ starting with $\tt{0}^n$. 
Each child of a node is a left child.  As a result, the RCL-ordering of $T$ corresponds to a pre-order traversal of $T$.  Since the symbol changed between a parent and child always increases from parent to child, this ordering will list the nodes (the necklaces in $\bN_t(n,w)$) in colex order.    %To obtain the running time result, it is relatively straightforward to apply Theorem ?? from~\cite{concattree}.  
\end{proof}

By traversing the concatenation tree $T$ in RCL order we can generate $\mathcal{U}_t(n,w)$.  To do this efficiently, we dynamically compute the children of each node as we traverse $T$, starting from the root $\tt{0}^n$.
Consider a node $\alpha = \tt{a}_1\tt{a}_2\cdots \tt{a}_n$ with change index $c$ and weight $w'$.
If $w' = w$, then $\alpha$ has no children.  Let $\alpha_i$ denote the string $\alpha$ with the symbol at index $i$ incremented by 1. If $\alpha_i$ is not a necklace for $1 < i \leq c$ then clearly $\alpha_{i-1}$ is also not a necklace.  If $\tt{a}_c < k-1$, we test if $\alpha_c$ is a necklace. If it is not a necklace, then $\alpha$ has no children.
Otherwise, we test $\alpha_{c-1}$, then $\alpha_{c-2}$, and so on until we find the largest index $i\leq c$ such that $\alpha_i$ is not a necklace or $i=0$.  The children of $\alpha$ are thus $\alpha_{i+1}, \cdots , \alpha_{c-1}$, including $\alpha_c$ if $\tt{a}_c < k-1$.  Thus, starting from the root $\tt{0}^n$ we can dynamically create the tree $T$ and visit its nodes in RCL order.   Testing if a string is a necklace and also determining the length of the aperiodic prefix of a necklace can be computed in $\mathcal{O}(n)$ time~\cite{Booth}. Each successful necklace test can be assigned to the corresponding child, and if there is a failed necklace test it can be assigned to the current node.  Thus, the nodes can be visited in $\mathcal{O}(n)$ amortized time.  
A complete C implementation of this RCL traversal algorithm to construct $\mathcal{U}_t(n,w)$ is given in the Appendix.

\begin{theorem}
$\mathcal{U}_t(n,w)$ can be constructed in $\mathcal{O}(1)$ amortized time per symbol using $\mathcal{O}(nt)$ space.
\end{theorem}
\begin{proof}
Recall from Lemma~\ref{lem:props} that the parent of every non-root periodic necklace is both unique and aperiodic. Discounting the root, this means that the number of periodic necklaces is less than or equal to the number of aperiodic necklaces. Since the aperiodic prefix of each aperiodic necklace contains $n$ symbols, this implies that $|\mathcal{U}_t(n,w)| \geq \frac{n}{2}|\bN_t(n,w)-\{\tt{0}^n\}|$.  Since each necklace in the concatenation tree $T$ is visited in $\mathcal{O}(n)$ amortized time, this implies that $\mathcal{U}_t(n,w)$ can be generated in $\mathcal{O}(1)$ amortized time per symbol.  
When visiting each node $\alpha$, only a constant amount of extra memory is required if we update and restore the values in $\alpha$ to test if a given $\alpha_i$ is a necklace and to visit each child. Since the depth of the recursion is bounded by $nt$, this construction requires $\mathcal{O}(nt)$ space for the run-time stack.
\end{proof}

%Using a recursive algorithm that builds strings from right to left, it is possible to generate the desired necklaces in colex order using only $\mathcal{O}(n)$ space using similar techniques as given in X for the binary case.

\subsection{Fixed weight and weight range}

Some applications may require both an upper and lower bound on the weights.
Binary \defo{fixed weight de Bruijn sequences} can be constructed by considering a shorthand representation~\cite{fixedweight2}. The shorthand strings correspond to all binary strings of length $n-1$ with weight in the range $[w-1,w]$.  This result is applied in ~\cite{dbrange} to construct 
binary \defo{weight-range de Bruijn sequences} for binary strings with weight in an arbitrary range $[w_1,w_2]$. Both constructions generate the sequences in $\mathcal{O}(1)$ amortized time per symbol using $\mathcal{O}(n)$ space.  There is no published construction for weight-range de Bruijn sequences over non-binary alphabets. As a special case, fixed-weight de Bruijn sequences for strings in $\Sigma_t(n)$ with weight exactly $w$ using a shorthand representation correspond to weight-range de Bruijn sequences for strings in $\Sigma_t(n-1)$ with weight in the range $[max(0,w-t+1), w]$.  If $w < t$, then such sequences correspond to bounded-weight de Bruijn sequences with an upper bound on the weight of $w$. 

\begin{theorem}
The universal cycle $\mathcal{U}_t(n,w)$ is a fixed-weight de Bruijn sequence using shorthand representatives when $w < t$.  Moreover, the sequence $\mathcal{U}_t(n,w)$ with the $\tt{0}^n$ substring replaced with $\tt{0}^{n-1}$ is a fixed-weight de Bruijn sequence using shorthand representatives when $w = t$.
\end{theorem}

\subsection{Applying the MSR} \label{sec:msr}

In this section, we present another bounded-weight de Bruijn sequence construction for $\Sigma_t(n,w)$ with the added constraint that $w < t$. Recall, the strings in $\Sigma_t(n,w)$ can be thought of as shorthand representatives for the subset of $\Sigma_t(n+1)$ containing strings with fixed weight $w<t$. Thus, when considering such a set $\Sigma_t(n,w)$, the  missing symbol $\tt{z}$ for a given string in the set is well defined.
Recall also that the MSR partitions $\Sigma_t(n,w)$ into equivalence classes corresponding to necklaces in $\bN_t(n+1)$.

\begin{exam}  \label{ex:msr}
Consider $\Sigma_5(3,4)$.  The MSR partitions this set of $35$ strings into the following 10 equivalence classes (columns):

\begin{center}
\begin{tabular}{c@{\hskip 2em}c@{\hskip 2em}c@{\hskip 2em}c@{\hskip 2em}c@{\hskip 2em}c@{\hskip 2em}c@{\hskip 2em}c@{\hskip 2em}c@{\hskip 2em}c}
   $\tt{\blue{0}00}$ &   $\tt{\blue{0}01}$ &  $\tt{\blue{0}10}$ &  $\tt{\blue{0}02}$ &  $\tt{\blue{0}11}$ &  $\tt{\blue{0}20}$  &  $\tt{\blue{0}03}$ &  $\tt{\blue{0}12}$ &  $\tt{\blue{0}21}$ &  $\tt{\blue{1}11}$ \\ 
   $\tt{\blue{0}04}$ &  $\tt{\blue{0}13}$ &  $\tt{\blue{1}03}$ &  $\tt{\blue{0}22}$ &  $\tt{\blue{1}12}$ &  $\tt{\blue{2}02}$   &  $\tt{\blue{0}31}$ &  $\tt{\blue{1}21}$ &  $\tt{\blue{2}11}$ &    \\
   $\tt{\blue{0}40}$ &  $\tt{\blue{1}30}$ &  $\tt{\blue{0}30}$ &  $\tt{\blue{2}20}$ &  $\tt{\blue{1}20}$ &                      &  $\tt{\blue{3}10}$ &  $\tt{\blue{2}10}$ &  $\tt{\blue{1}10}$ &   \\
   $\tt{\blue{4}00}$ &  $\tt{\blue{3}00}$ &  $\tt{\blue{3}01}$ &  $\tt{\blue{2}00}$ &  $\tt{\blue{2}01}$ &                      &  $\tt{\blue{1}00}$ &  $\tt{\blue{1}01}$ &  $\tt{\blue{1}02}$ &    \\
\end{tabular}
\end{center}
 Each equivalence class (cycle) can be represented by a necklace in $\bN_5(4)$ with weight $w=4$.  Each necklace is highlighted by reading the first symbol down each column.  
 
 %The corresponding first-non zero cycle joining tree is illustrated below.

%\includegraphics[scale=1.10]{trees/43MSR.pdf}
 
%Observe that the labeled pre-order traversal of this tree visits the necklaces (nodes) in \emph{reverse} colex order.%

\end{exam}

%

%

%Note, they have been considered in the past for binary strings but without this extra constraint on the value of the weight~\cite{fixedweight2}. 

Using the MSR, cycle-joining trees and successor rules can be created in a similar manner as we did when applying the PCR using the \emph{first non-zero} parent rule.  In this case, the nodes have length $n+1$. By decrementing the first non-zero symbol the following symbol in the parent necklace is incremented based on the MSR, maintaining the weight constraint. 
%
%=================
\begin{result} \noindent
{\bf First non-zero (MSR)} (with root $\tt{0}^n\tt{w}$): the parent of non-root node $\tt{a}_1\tt{a}_2\cdots \tt{a}_{n+1}$, where $j$ denotes the smallest index such that $\tt{a}_j \neq 0$, is 
%$\tt{a}_1\cdots \tt{a}_{j-1}(\tt{a}_j{-}1)(\tt{a}_{j+1}{+}1)a_{j+2}\ldots \tt{a}_{n+1}$.
$\tt{0}^{j-1}(\tt{a}_j{-}1)(\tt{a}_{j+1}{+}1)a_{j+2}\ldots \tt{a}_{n+1}$.
\end{result}
%=================
%
\noindent
This parent rule is well-defined since the weight constraint implies that  $j < n+1$ and $\tt{a}_{j+1} < t-1$. 
As an example, see the cycle-joining tree for $\Sigma_k(n,w)$ in Figure~\ref{fig:treeMSR}.

%=================
\begin{figure} [h]
\begin{center}
  \includegraphics[scale=0.9]{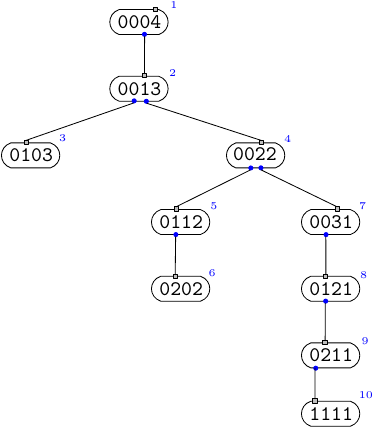}
  \caption{ The first non-zero cycle-joining tree based on the MSR for $\Sigma_5(3,4)$ where the nodes are represented by the necklaces in $\bN_5(4)$ with weight exactly $w=4$.  The labeled pre-order traversal of this tree visits the necklaces in \emph{reverse} colex order. The tree is drawn in the style of a concatenation tree even though they are not formally defined for MSR-based cycle-joining trees.}
  \label{fig:treeMSR}
\end{center}
\end{figure}
%=================

The following is a space-efficient successor rule based on the functions $g(\alpha)$ and $h(\alpha)$ (recall its definition in Section~\ref{sec:cycle-join}) for this cycle-joining tree, obtained by considering the corresponding conjugate pairs. 

%========================================================
% MSR SUCCESSOR RULE
%========================================================
\begin{result} 
\noindent {\bf First non-zero (MSR)  successor rule} for $\alpha = \tt{a}_1\tt{a}_2\cdots \tt{a}_{n}$ with missing symbol $\tt{z}$ 

\smallskip

\noindent
Let $j$ be the largest index of $\tt{a}_2\tt{a}_3\cdots \tt{a}_{n}$ such that $\tt{a}_j \ne \tt{0}$ or $j=1$ if no such index exists.
Let $\tt{x}$ be the largest symbol in $\{\tt{1,2, \ldots, t{-}1}\}$  such that (i) $\tt{y} = \tt{z}-\tt{x}+\tt{a}_1$ is in $\Sigma_t$ and  (ii) $\tt{0}^{{n}-j}   \tt{x} \tt{y} \tt{a}_2\cdots \tt{a}_{j}$ is a necklace,
or let $\tt{x}=-1$ if no such symbol exists.

\begin{center}  \small
$h_2(\alpha) = \left\{ \begin{array}{ll}
          0 		&\ \   \mbox{if  $\tt{x} \ne -1$ and $\tt{z}  = \tt{x}$;}\\
         \tt{z} {+}1	&\ \  \mbox{if  $\tt{x} \ne -1$ and $\tt{z} < \tt{x}$;}   \ \ \ \ \ \ \ \ \\
         \tt{z} \  	&\ \  \mbox{otherwise.}\end{array} \right.$
\end{center}

\vspace{-0.1in}
\end{result}
%========================================================

\noindent
Let $\mathcal{V}_t(n,w)$ denote the universal cycle for $\Sigma_t(n,w)$ obtained by starting with $\tt{0}^n$ and repeatedly applying the above successor rule $|\Sigma_t(n,w)|- n$ times on the last $n$ symbols to obtain the next symbol in the sequence.  Na\"{i}vely, in order to determine $\tt{x}$ the function $h_2(\alpha)$ will require testing up to $t$ strings to see whether they are a necklace.  
%Testing whether a string is a necklace can be done in $\mathcal{O}(n)$ time~\cite{Booth}.  
However, with some fairly straightforward optimizations, we can reduce the number of necklace tests to one.

\begin{theorem} \label{thm:MSR}
For $n,t \geq  2$ and $0 \leq w < t$, the universal cycle $\mathcal{V}_t(n,w)$ for $\Sigma_t(n,w)$ can be constructed via the successor rule $h_2(\alpha)$ in $\mathcal{O}(n)$ time per symbol using $\mathcal{O}(n)$ space.
\end{theorem}
\begin{proof}  
We demonstrate that the value $\tt{x}$ in the successor rule can be determined by performing at most one necklace test.  Let $\beta = \tt{0}^{{n}-j}   \tt{x} \tt{y} \tt{a}_2\cdots \tt{a}_{j}$.
Let $\tt{x}'$ denote the smallest symbol following a substring $\tt{0}^{{n}-j}$ in $\tt{a}_2\cdots \tt{a}_{j}$, or let $\tt{x}'=\tt{t-1}$ if no such substring exists. In order for $\beta$ to be a necklace, clearly $\tt{x} \leq \tt{x}'$. 
Furthermore, $\tt{x} + \tt{y} = \tt{a}_1+\tt{z}$ to maintain the weight constraint, which means $\tt{x} \leq \tt{a}_1 + \tt{z}$.  Consider $\tt{x}$ to be the minimum of  $\tt{x}'$ and $\tt{a}_1 + \tt{z}$.  If this value is 0, then no necklace test is required.  If $\beta$ is a necklace for this value of $\tt{x}$, no more tests are required as it is the largest possible value for $\tt{x}$.  
Otherwise, decrementing the value of $\tt{x}$ (and incrementing $\tt{y}$) will result in a necklace since the length $n-j+1$ prefix of $\beta$ will be smaller than all other substrings of $\beta$ (considered cyclically) of the same length, by the values of $\tt{x}'$ and the fact that $\tt{y} \neq 0$.
\end{proof}

The following result follows from the fact that the strings in $\Sigma_t(n,w)$ can be thought of as shorthand representatives for the subset of $\Sigma_t(n+1)$ containing strings with fixed weight $w<t$.

\begin{theorem}
The universal cycle $\mathcal{V}_t(n,w)$ is a fixed-weight de Bruijn sequence using shorthand representatives when $w < t$.  Moreover, the sequence $\mathcal{V}_t(n,w)$ with the $\tt{0}^n$ substring replaced with $\tt{0}^{n-1}$ is a fixed-weight de Bruijn sequence using shorthand representatives when $w = t$.
\end{theorem}

Interestingly, $\mathcal{V}_t(n,w)$ also appears to have interesting concatenation properties.  Consider:
\[ \mathcal{V}_5(3,4) = \tt{\underline{000}4 \cdot 0013 \cdot 0103 \cdot 0022 \cdot 0112 \cdot 02 \cdot 0031\cdot 0121 \cdot 0211 \cdot 1.} \]
This sequence corresponds to the concatenation of the aperiodic prefixes of the necklaces  in $\bN_5(4)$  with exactly weight $w=4$ as they appear in \emph{reverse} colex order!  Moreover, observe that a preorder traversal of the corresponding ``concatenation-like tree'' (see Figure~\ref{fig:treeMSR}) also corresponds to visiting the necklaces in reverse colex order.
These observations lead to the following conjecture, where $\mathcal{V'}_t(n,w)$ denotes the sequence obtained by 
concatenating the aperiodic prefixes of the necklaces in $\bN_t(n+1)$  with weight $w$ as they appear in \emph{reverse} colex order.
\begin{conjecture}
The universal cycle $\mathcal{V}_t(n,w)$ is equivalent to $\mathcal{V'}_t(n,w)$, where $n,t \geq 2$ and $w < t$.   Moreover, 
the universal cycle can be generated in $\mathcal{O}(1)$ amortized time per symbol.
\end{conjecture}

\noindent
We believe that this conjecture can be proved by generalizing the (PCR-based) concatenation tree framework~\cite{concattree} to other feedback functions (like the MSR), then following a similar analysis to the one we gave for the sequence $\mathcal{U}_t(n,w)$ to obtain the running time result.

\section{Application to $k$-subsets and $k$-multisets}  \label{sec:results}

In this section, we apply the results from the previous section to $k$-subsets assuming $n \geq k \geq 1$, and $k$-multisets assuming $n, k \geq 2$.  Recall the following observations made in Section~\ref{sec:intro}.
\begin{itemize}
\item When the subsets of $\bS_k(n)$ are represented by their difference representatives with each symbol $\tt{x}$ mapped to $\tt{x{-}1}$, the subsets correspond to $\Sigma_{n-k+1}(k,n-k)$.
\smallskip

\item When the multisets of $\bM_k(n)$ are represented by their shorthand frequency representatives, the multisets correspond to $\Sigma_{k+1}(n-1,k)$. \smallskip

\item When the multisets of $\bM_k(n)$ are represented by their difference representatives, the multisets correspond to $\Sigma_{n}(k,n-1)$.

%By using the difference representation and then mapping each symbol , the set of subsets $\bS_k(n)$ can be represented by the set of strings over the  alphabet $\{0,1,\ldots , n{-}k\}$ of length $k$, where the weight of each string is bounded above by $n{-}k$.  
%\item By using the shorthand frequency representation, the set of multisets $\bM_k(n)$ can be represented by the set of strings over the alphabet $\{\tt{0,1},\ldots, \tt{k}\}$ of length $n-1$, where the weight of each string is bounded above by $k$. 
%\item By using the difference representation, the set of multisets $\bM_k(n)$ can be represented by the set of strings over the alphabet $\{\tt{0,1},\ldots, \tt{n-1}\}$ of length $k$, where the weight of each string is bounded above by $n-1$. 
\end{itemize}
 %
 %In all three cases, the weight constraint is the same as the size of the alphabet. %This property will be useful later in this section.

\noindent
These observations immediately give rise to the following three results. 

\begin{theorem}
%A universal cycle for $\mathcal{S}_1$ for $\bS_k(n)$ using difference representatives can be constructed in $\mathcal{O}(1)$ amortized time per symbol.  Moreover, given any subset in $\bS_k(n)$ with difference representative $d_1d_2\cdots d_k$, the symbol following this string in $\mathcal{S}_1$ can be computed in $\mathcal{O}(n)$ time. 
$\mathcal{U}_{n-k+1}(k,n-k)$ and $\mathcal{V}_{n-k+1}(k,n-k)$ are universal cycles  
for $\bS_k(n)$ using difference representatives, where each symbol $\tt{x}$ in the representative is mapped to $\tt{x-1}$.  
\end{theorem}

%======================
\begin{exam}
Consider two universal cycles for $\Sigma_4(3,3)$:
\[ \mathcal{U}_4(3,3) = \tt{0  \cdot 001 \cdot  011  \cdot 1  \cdot 021  \cdot 002 \cdot  012  \cdot 003, \mbox{ and }} \]
\[ \mathcal{V}_4(3,3) = \tt{0003 \cdot 0012 \cdot 0021 \cdot 0111 \cdot 0201.} \]
By mapping each symbol $\tt{x}$ to $\tt{x+1}$ in the above sequences, we obtain universal cycles $\mathcal{S}_1$ and $\mathcal{S}_2$ for $\bS_3(6)$ using difference representatives:
%where each substring of length $3$ corresponds to a unique difference representative for $\bS_3(6)$:
%
\[ \mathcal{S}_1 =  \tt{1  \cdot 112  \cdot 122 \cdot  2  \cdot 132  \cdot 113  \cdot 123  \cdot 114}, \mbox{ and } \]
\[ \mathcal{S}_2 =  \tt{1114 \cdot 1123 \cdot 1132 \cdot 1222 \cdot 1312.}  \]
Each universal cycle above has the expected length of  ${6 \choose 3} = 20$.
\end{exam}

\begin{theorem} 
%A universal cycle for $\mathcal{M}_1$ for $\bM_k(n)$ using  shorthand frequency representatives can be constructed in $\mathcal{O}(1)$ amortized time per symbol.  Moreover, given any subset in $\bM_k(n)$ with shorthand frequency representative $f_1f_2\cdots f_{n-1}$, the symbol following this string in $\mathcal{M}_1$ can be computed in $\mathcal{O}(n)$ time.  
$\mathcal{U}_{k+1}(n-1,k)$ and $\mathcal{V}_{k+1}(n-1,k)$ are  universal cycles  
for $\bM_k(n)$ using shorthand frequency  representatives.  
\end{theorem}

\begin{theorem} 
%A universal cycle for $\mathcal{M}_2$ for $\bM_k(n)$ using  difference representatives can be constructed in $\mathcal{O}(1)$ amortized time per symbol.  Moreover, given any subset in $\bM_k(n)$ with difference representative $d_1d_2\cdots d_{k}$, the symbol following this string in $\mathcal{M}_2$ can be computed in $\mathcal{O}(n)$ time.  
$\mathcal{U}_{n}(k,n-1)$  and $\mathcal{V}_{n}(k,n-1)$ are universal cycles  
for $\bM_k(n)$ using difference representatives.  
\end{theorem}

\noindent
%By applying Theorem~\ref{thm:main} and the previously mentioned fact that the First non-zero (Grandmama) successor rule runs in $\mathcal{O}(n)$-time, we obtain the following results.

\begin{exam}
The universal cycles $\mathcal{U}_5(3,4)$ and $\mathcal{V}_5(3,4)$ are universal cycles for 
$\bM_4(4)$ using shorthand frequency representatives:
\[ \mathcal{U}_5(3,4) = \tt{0 \cdot 001 \cdot 011 \cdot 1 \cdot 021 \cdot 031 \cdot 002 \cdot 012 \cdot 112 \cdot 022 \cdot 003 \cdot 013 \cdot 004}, \mbox{ and } \]
\[ \mathcal{V}_5(3,4) = \tt{0004 \cdot 0013 \cdot 0103 \cdot 0022 \cdot 0112 \cdot 02 \cdot 0031\cdot 0121 \cdot 0211 \cdot 1.} \]

\smallskip

\noindent
The universal cycles $\mathcal{U}_4(4,3)$ and $\mathcal{V}_4(4,3)$ are universal cycles for $\bM_4(4)$ using difference representatives:
\[ \mathcal{U}_4(4,3) =   \tt{0 \cdot 0001 \cdot  01 \cdot  0011 \cdot 0111 \cdot 0021 \cdot 0002 \cdot 0102 \cdot 0012 \cdot 0003}, \mbox{ and }\]
\[ \mathcal{V}_4(4,3) =   \tt{00003 \cdot 00012 \cdot  00102 \cdot  00021 \cdot 00111 \cdot 01011 \cdot 00201}.\]
\noindent
Each universal cycle above has the expected length of  ${4+4-1 \choose 4} = 35$.

\end{exam}

\noindent
The above three results imply that universal cycles for $k$-subsets and $k$-multisets can be constructed in $\mathcal{O}(1)$ amortized time per symbol, or via an $\mathcal{O}(n)$ time per symbol successor rule.  %Either approach requires $\mathcal{O}(n)$ space.    
The latter two results are the first known universal cycle constructions for $k$-multisets.

\section{Acknowledgment}
Joe Sawada (grant RGPIN-2025-03961) gratefully acknowledges research support from the Natural Sciences and Engineering Research Council of Canada (NSERC).

%-----------  BIBLIOGRAPHY  ---------------

\bibliographystyle{alphaabbr}
\bibliography{refs}
\normalsize
\appendix

\newpage
\section{C implementation of the bounded weight ``Grandmama'' de Bruijn sequence}

\footnotesize

\begin{code}
//===========================================================================
// Genereate the (upper) bounded weight "Grandmama" de Bruijn sequence 
// over alphabet {0,1, ... , t-1} and upper bound on weight of w.
// It can be applied to construct a UC for k-subsets or k-multisets of [n].
//===========================================================================
#include<stdio.h>
int a[100],n,t,w;

//=============================================================================
// Test if a[1..n] is a necklace, if so return the length of its aperiodic
// prefix; otherwise return 0
//===========================================================================
int IsNecklace(int a[]) {
    int i,p=1;
    
    for (i=1; i<=n; i++) {
        if (a[i] < a[i-p]) return 0;
        if (a[i] > a[i-p]) p=i;
    }
    if (n%p == 0) return p;
    return 0;
}
//=============================================================================
// Visit "first non-zero" concatenation tree in RCL order by dynamically
// generating children of the current node a[1..n] with aperiodic prefix of
// length p, change index c, and weight w2.
// This visits necklaces with bounded weight w in colex order.
//=============================================================================
void RCL(int a[], int p, int c, int w2) {
    int i,j;
    
    // VISIT: Print aperiodic prefix of a[1..n]
    for (i=1; i<=p; i++) printf("%d", a[i]);
    if (w2 == w) return;    // No children when max weight achieved
    
    // Scan from c left to determine the first index for a child
    a[c]++;
    if (a[c] < t && !IsNecklace(a))  { a[c]--; return; }
    a[c]--;
    
    j=c-1;
    a[j] = 1;
    while (j >=1 && IsNecklace(a)) {
        a[j] = 0; j--; a[j] = 1;
    }
    a[j] = 0;
    
    // Visit children from left to right, handle change index separately
    for (i=j+1; i<c; i++) {
        a[i] = 1;
        RCL(a, IsNecklace(a), i, w2+1);
        a[i] = 0;
    }
    if (a[c] < t-1) {
        a[c]++;
        RCL(a, IsNecklace(a), c, w2+1);
        a[c]--;
    }
}
//===============================================
int main() {
    
    printf("Enter n t w: ");  scanf("%d %d %d", &n, &t, &w);
    for (int i=1; i<=n; i++) a[i] = 0;
    RCL(a,1,n,0);
}
    
\end{code}

\end{document}